\documentclass[letterpaper]{article}
\usepackage{aaai}
\usepackage{times}
\usepackage{helvet}
\usepackage{courier}

\usepackage{color}
\usepackage{algorithm}
\usepackage{algorithmic}
\usepackage{amsmath}
\newtheorem{theorem}{Theorem}
\newtheorem{lemma}[theorem]{Lemma}

\usepackage{bm}
\usepackage{subfigure}
\usepackage{graphicx}
\usepackage{epstopdf}
\usepackage{color} % for revising use only

\begin{document}
\title{Incentivizing High-quality Content from Heterogeneous Users: \\On the Existence of Nash Equilibrium}
\author{Yingce Xia$^*$, Tao Qin$^\dag$, Nenghai Yu$^\ddag$ \and Tie-Yan Liu$^\dag$\\
$^{*\ddag}$Key Laboratory of Electromagnetic Space Information of CAS, USTC, Hefei, Anhui, 230027, P.R. China\\
$^\dag$Microsoft Research, Building 2, No. 5 Danling Street, Haidian District, Beijing, 100080, P.R.China\\
$^*$yingce.xia@gmail.com, $^\dag$\{taoqin,tyliu\}@microsoft.com, $^\ddag$ynh@ustc.edu.cn
}
\maketitle

\def\QEDclosed{\mbox{\rule[0pt]{1.3ex}{1.3ex}}}
\def\QEDopen{{\setlength{\fboxsep}{0pt}\setlength{\fboxrule}{0.2pt}\fbox{\rule[0pt]{0pt}{1.3ex}\rule[0pt]{1.3ex}{0pt}}}}
\newenvironment{proof}[1][Proof.]{\begin{trivlist}
\item[\hskip \labelsep {\bfseries #1}]}{\QEDclosed\end{trivlist}}
\newenvironment{proofsketch}[1][Proof sketch.]{\begin{trivlist}
\item[\hskip \labelsep {\bfseries #1}]}{\QEDclosed\end{trivlist}}

\begin{abstract}
\begin{quote}
We study the existence of pure Nash equilibrium (PNE) for the mechanisms used in Internet services (e.g., online reviews and question-answer websites) to incentivize users to generate high-quality content. Most existing work assumes that users are homogeneous and have the same ability. However, real-world users are heterogeneous and their abilities can be very different from each other due to their diverse background, culture, and profession. In this work, we consider heterogeneous users with the following framework: (1) the users are heterogeneous and each of them has a private type indicating the best quality of the content she can generate; (2) there is a fixed amount of reward to allocate to the participated users.  Under this framework, we study the existence of pure Nash equilibrium of several mechanisms composed by different allocation rules, action spaces, and information settings. We prove the existence of PNE for some mechanisms and the non-existence of PNE for some mechanisms. We also discuss how to find a PNE for those mechanisms with PNE either through a constructive way or a search algorithm.
\end{quote}
\end{abstract}

\newcommand{\myeqref}[1]{Eqn. \eqref{#1}}

\section{Introduction}
More and more Internet websites rely on users' contribution to collect high-quality content, including knowledge-sharing services (e.g., Yahoo! Answers and Quora) , online product commenting and rating services (e.g., Yelp, mobile app stores), and ecommerce websites (e.g., Amazon.com). For simplicity, we call websites that rely on User-Generated Content  UGC websites. To attract more users and to incentivize them to contribute high-quality content, those sites usually give high-quality contributors some reward in the form of virtual value, which represents privilege and benefit, or monetary return, such as the gift card. To collect more reward, users usually strategically interact with those websites. Therefore, to maximize the quality of the content generated from users, a UGC website needs to carefully design their mechanisms and analyze users' behaviors. We call the mechanisms used by those UGC sites UGC mechanisms.

Recently, much effort has been devoted to the design and analysis of UGC mechanisms  \cite{Ghosh:2011:GAR:1993574.1993603,Ghosh:2011:IHU:1963405.1963428,Easley:2013:IGG:2492002.2482571}. Most of those works assume that users are homogeneous - people are of the same ability while contributing to the sites. However, in the real world, users' abilities can be very different from each other due to their diverse background, culture, and profession. For example, an experienced photographer can write a high-quality comment to a photo, which is very difficult for a non-experienced user. Thus, in this work, we study the game theoretical problem raised in Internet services with heterogeneous users. We introduce the concept of ``type'' for the problem,  which denotes the ability of a user: the larger the type of a user is, the better content she can contribute to the site. We further assume that each user needs to afford a cost to participate in the game and contribute content. The cost reflects the effort of content generation (e.g. writing a review), such as time and  mobile traffic. In our work, we assume users costs are bounded which is different from \cite{Easley:2013:IGG:2492002.2482571,Ghosh:2011:IHU:1963405.1963428}. We believe that our setting is more practical because nobody would spend too much (if not infinite) effort to make contributions.

Two allocation rules are studied in our work: one is the top $K$ allocation rule\cite{Jain:2009:DIO:1566374.1566393,Easley:2013:IGG:2492002.2482571} in which users with the highest $K$ qualities will get the reward equally. The other rule is proportional-share rule\cite{Ghosh:2011:IHU:1963405.1963428,Jain:2009:DIO:1566374.1566393,Nisan:2007:AGT:1296179,chen2009essays}, in which all the participants (who make non-zero contribution) will share the reward proportionally to their contributed qualities. The proportional share rule is also widely used in network rate control\cite{kelly1997charging,kelly98ratecontrol}, market allocation\cite{cachon1999equilibrium} and scheduling\cite{sandholm2010dynamic,563725}. Two action spaces are investigated: the binary action space, in which each user can only choose to participate in or not; and the continuous action space, in which each user can choose the quality of the content to contribute. Besides, we study the problem from both the full-information setting and the partial-information setting.

We study the existence of pure Nash equilibrium for several different UGC mechanisms by combining the above options.   Our main results can be summarized as follows.
\begin{enumerate}
\item For the \emph{full-information setting}, we prove the existence of PNE for the mechanism with the proportional allocation rule and the continuous action space. The key of the proof is to construct a perturbed game, prove the existence of PNE for the perturbed game, and prove that the PNE of the perturbed game will converge to the equilibrium of the original game. We then discover several properties of the PNE of the mechanism, which are further used to design an algorithm to find a PNE for the mechanism.  We also study three other mechanisms and show (1) the existence of PNE for the mechanism with the top $K$ allocation rule and the binary action space and for the mechanism with the proportional allocation rule and the binary action space and (2) the non-existence of PNE for the mechanism with the top $K$ allocation rule and the continuous action space.
\item For the \emph{partial-information setting}, we prove the existence of a symmetric PNE for the mechanism with the top $K$ allocation rule and the continuous action space. The key of the proof is to construct a simple but (maybe) infeasible symmetric strategy and then convert it to a feasible symmetric equilibrium strategy by repeated calibration. Our proof also provides a method to construct a symmetric PNE.  For the binary action space, we prove the existence of equilibrium for the mechanisms with both top $K$ allocation rule and proportional allocation rule.
\end{enumerate}

\section{Related work}
Recently UGC mechanisms have attracted much research attention \cite {Anderson:2013:SUB:2488388.2488398,Chawla:2012:OCC:2095116.2095185,Ghosh:2011:IHU:1963405.1963428,Ghosh:2011:GAR:1993574.1993603,Ghosh:2013:IPO:2492002.2482587}. Most of the existing work concentrates on homogeneous settings, i.e. users are of the same ability to generate content for the website.
  \cite{Ghosh:2011:IHU:1963405.1963428} designs a simple voting rule under sequential and simultaneous model, in which both the quality of contributions and the number of contributors are endogenously determined. \cite{Ghosh:2011:GAR:1993574.1993603} studies  the rank based allocation mechanism for websites with user-generated content and shows the mechanism always incentivizes higher quality equilibriums than the proportional allocation rule.  \cite{Ghosh:2013:IPO:2492002.2482587} models the online education forums with two parameters which represent the frequency of checking forums by teachers and students separately. A brief survey about UGC mechanisms can be found in \cite{Ghosh:2012:SCU:2509002.2509006}.

  The key differences between those existing literature and our work are that (1) we focus on heterogeneous users who have different abilities to generate content, and (2) we assume users' effort is bounded, i.e., users cannot afford an infinity cost to make contributions. One closely related work is \cite{Easley:2013:IGG:2492002.2482571}, which studies the problem of badge design. Although users' abilities are considered in that work, the participants are modeled as a continuum and as a result, the single user's behavior will not affect others' payoff much. In this work, we regard users as discrete individuals, and one's strategy will impact others' utilities.

\section{Model}
In this section, we describe the model for analyzing the incentives created by various UGC mechanisms, when contributors are strategic agents with heterogeneous abilities, and when the decision of whether to participate in and how much to contribute is a strategic choice.

There is a set of $N$ strategic users in a UGC site, and each user $i$ has a private type $q_i\in[0,1]$, which indicates the best quality of the content the user can contribute to the site. Without loss of generality, we number the users according to the descending order of their types, i.e. $q_1 \geq q_2 \ldots \geq q_N$.
Let $x_i$ denote user $i$'s action, which indicates the quality of the content she actually contributes to the site. Note that we have $0\le x_i \le q_i$. The user needs to afford a cost $c_i$ for the action $x_i$. In this work, we consider linear cost for simplicity: $c_i=c\frac{x_i}{q_i},$ where $c$ is a same upper bound\footnote{If each user has a different cost upper bound $C_i$, it is easy to absorb $C_i$ into the private type $q_i$ by scaling: $q_i=\frac{q_ic}{C_i}$. } of the cost that a user can afford.

We study two action spaces in this work. The first one is a binary action space: each user can only choose not to contribute or to contribute content with quality $q_i$ (i.e, $x_i\in \{0, q_i\}$). The second one is a continuous action space: the quality $x_i$ that user $i$ contributes to the site is a continuous value between $0$ and $q_i$ (i.e., $x_i\in [0,q_i]$). Note sometimes we say that a user does not participate in the game if $x_i=0$, and a user participates in the game if $x_i>0$.

The site has a fixed number of reward $R$ to allocate to the contributors, depending on their contributions. We study two allocation rules: the top $K$ allocation \cite{Jain:2009:DIO:1566374.1566393} and the proportional allocation \cite{Ghosh:2011:IHU:1963405.1963428}. The first one allocates $\frac{R}{K}$ to each user of those who contribute the top $K$ largest qualities. Note that if $N<K$, each user can still only get $\frac{R}{K}$ reward. The second one allocates the reward to all users proportional to their contributions: the reward $r_i$ allocated to user $i$ is $\frac{x_i}{\sum_jx_j}R$ if $x_i > 0$ and $0$ if $x_i=0$.

While analyzing the model, we consider two settings: the full-information setting and the partial-information setting. In the full-information setting, the types $\{q_i\}_{i\in[N]}$ are deterministic and are known to all the users. In the partial-information setting, the type of each user is assumed to be drawn from a publicly known distribution $F$, the first order derivative of which is continuous, and each user only knows her own type $q_i$.

With the above notations, the utility of user $i$ can be written as $u_i(x_i,x_{-i})=r_i(x_i,x_{-i})-c_i(x_i),$ where $r(x_i,x_{-i})$ is the reward of user $i$ given her strategy $x_i$ and the strategies $x_{-i}$ of other players. We assume that all the users are rational and they want to maximize their (expected) utilities.

\section{Full-information Setting}
In this section, we study the mechanisms under the full-information setting. Recall that in this setting, the type $q_i$ of any user is known to all the users. This setting corresponds to the real-world scenarios where the users are familiar with each other. For example, considering a professional mathematical question posted in Yahoo! Answer, there will be only a few users in the Yahoo! Answer community who can answer the question and they know each other quite well.

Combining the different choices of the allocation rule and the action space, there are four mechanisms under this setting. We focus on the mechanism with the proportional allocation rule and the continuous action space here and directly list the results of the other three, which are relatively easy to analyze.

\subsection{$\mathcal{M}_1$: Top $K$ Allocation, Binary Action Space}
It is not difficult to see that PNE exists and is unique (except the case $R=Kc$) for this scheme. There are three cases depending on the parameters ($R, K, c$):
\begin{enumerate}
\item If $R<Kc$, no user will contribute content to the site and the equilibrium is $x_i=0, \forall i$.
\item If $R=Kc$, there are many equilibria: any group of $k\le K$ users contributing to the site is an equilibrium.
\item if $R> Kc$, the first $K$ users will contribute to the site. That is, $x_i=q_i, \forall 1\le i \le K$ and $x_i=0, \forall i>K$.
\end{enumerate}

\subsection{$\mathcal{M}_2$: Top $K$ Allocation, Continuous Action Space}
There does not exist PNE for this mechanism under the full-information setting according to the following discussions.
\begin{enumerate}
    \item Zero-participating and more than $K$ people participating in are obviously not an equilibrium.
    \item If the equilibrium is constructed by fewer than $K$ people participating in, there is at least one person would could get positive utility by making a positive contribution. This contradicts with the concept of equilibrium.
    \item If there are $K$ people participating in, all the contributors will generate contents with quality $\epsilon \rightarrow 0$ but $\epsilon \ne 0$, so the equilibrium strategy does not exist.
\end{enumerate}

\subsection{$\mathcal{M}_3$: Proportional Allocation, Binary Action Space}
It turns out that PNE exists and there can be multiple equilibria for this mechanism.
\begin{enumerate}
\item If $R<c$, there exists a unique equilibrium in which nobody will contribute: $x_i=0, \forall i$.
\item If $R=c$, multiple equilibria exist: (1) nobody contributing is an equilibrium, and (2) any single user contributing is also an equilibrium.
\item If $R>c$, Nash equilibrium exists. Denote $j$ as the index satisfying the following two inequalities:
$$\frac{R q_j}{\sum_{k=1}^{j}q_k} > c,$$
and $$\frac{R q_{j+1}}{\sum_{k=1}^{j+1}q_k} \leq c.$$
Then $x_i=q_i, \forall 1\le i \le j$ and $x_i=0, \forall i>j$ compose an equilibrium.
There can exist multiple equilibria. Consider an example with parameters $N=3,R = 4,c = 1,\bm{q} = \{0.9247 , 0.3421,0.3095\}$. One can verify that both $\bm{x} = \{0.9247,0.3421,0\}$ and $\bm{x} = \{0.9247,0,0.3095\}$ are both equilibria.
\end{enumerate}

\subsection{$\mathcal{M}_4$: Proportional Allocation, Continuous Action Space}
We first prove the existence of PNE and then present an algorithm to search the PNE for the mechanism $\mathcal{M}_4$.
\begin{theorem}
  For the full-information setting, there exists a PNE for the mechanism $\mathcal{M}_4$.
  \label{thm:equiexist_var_por}
\end{theorem}
\begin{proof}
Consider an action profile $\{x_i\}_{i\in[N]}$. Denote $x_{-i} = \sum_{j \ne i}x_j$.  If $\sum_{i=1}^{N}x_i >0$, the utility of user $i$ is
\begin{small}
\begin{equation}
  u_i(x_i , x_{-i}) = R\frac{x_i}{x_i + x_{-i}} - \frac{x_i}{q_i}c,
\end{equation}
\end{small}
constrained by $0 \leq x_i \leq q_i$.

The first order derivative of $u_i$ w.r.t $x_i$ is
\begin{small}
\begin{equation}
  \frac{\partial{u_i}}{\partial{x_i}} = R\frac{x_{-i}}{(x_i + x_{-i})^2} - \frac{c}{q_i}.
  \label{eq:foc_utility}
\end{equation}
\end{small}
By setting the above derivative to zero, we get the best response strategy of user $i$:
\begin{small}
\begin{equation}
  x_i(q_i  , x_{-i}) = \sqrt{\frac{Rq_ix_{-i}}{c}} - x_{-i}.
  \label{eq:bstResponse3}
\end{equation}
\end{small}
We have the following observations for this best response strategy:
\begin{enumerate}
\item Clearly $x_i = 0$ is not the best response when $x_{-i} = 0$ because user $i$ would not be rewarded with $x_i=0$. Actually there is not a best response for user $i$ when $x_{-i}=0$:  if she gets a positive utility by contributing $\delta > 0$, she will profitably deviate by contributing $\frac{\delta}{2}$. Therefore zero-contribution ($x_i=0, \forall i$) is not an equilibrium strategy but it is a fixed point for \myeqref{eq:bstResponse3}.
\item $x_i$ calculated from \myeqref{eq:bstResponse3} can be smaller than zero or larger than $q_i$, which is not a feasible action. If $x_i(q_i,x_{-i}) < 0$, it means that \myeqref{eq:foc_utility} will be smaller than zero when $x_i > 0$, so it is better to make zero contribution. If $x_i(q_i,x_{-i}) \geq q_i$, \myeqref{eq:foc_utility} will be larger than $0$, so $u_i(x_i,x_{-i})$ increases w.r.t. $x_i$. Therefore it is better to contribute $q_i$.
\end{enumerate}
Based on the two observations, we consider a perturbed game \cite{journals/tpds/FeldmanLZ09}, in which the action of each user is lower bounded by a small positive quality $\epsilon$ and the best response strategy of $i$ is as below.
\begin{equation}
  x_i^*(q_i , x_{-i} , \epsilon) = \left\{
    \begin{aligned}
      &\epsilon &\textrm{if}\;x_i(q_i , x_{-i}) \leq \epsilon\\
      &x_i(q_i , x_{-i}) &\textrm{if}\;\epsilon <x_i(q_i , x_{-i}) <q_i\\
      &q_i &\textrm{if}\;x_i(q_i , x_{-i}) \geq q_i
    \end{aligned}
  \right.
  \label{eq:equiEquations}
\end{equation}
where $x_i(.,.)$ is defined in \myeqref{eq:bstResponse3}.
In the remaining part of the proof, we show that (1) there exists a PNE for the perturbed game; (2) as $\epsilon \rightarrow 0$, any PNE of the perturbed game will not converge to zero contribution point and (3) by setting $\epsilon \rightarrow 0$, we get the PNE for the original game.

Denote the space $[0,q_1] \times [0 , q_2] \times ... \times [0 , q_N]$ as $X$, which is convex and compact. Define a mapping $f$ from $X$ to itself, in which for any fixed $\epsilon$, $\forall i \in [N]$, $f_i(\bm{x},\epsilon) = x^*_i(q_i,x_{-i},\epsilon)$. It is easy to verify that $f$ is a continuous mapping. According to Brouwer fixed-point theorem\cite{Border1}, we know $f$ has at least one fixed-point in $X$, which is the equilibrium of the perturbed game. Denote one fixed point as\footnote{We denote vectors with bold face letters in this paper.} $\bm{x}^{\epsilon}$.  Since $X$ is compact, we could always find a series of $\{\epsilon_n\} \rightarrow 0$ with their corresponding $\bm{x}^{\epsilon_n}$ converging. Denote the limit point as $\bm{x}^0$.

Next we show $\bm{x}^0$ is not the zero contribution point. Otherwise, $\sum_{i=1}^{N}x_i^{\epsilon_n}$ could be infinitely close to zero as $n \rightarrow \infty$. What's more, both $x_1^{\epsilon_n}$ and $x_2^{\epsilon_n}$ are strictly less than the $q_1$ and $q_2$ respectively. We set $\sum_{i=3}^{N}x_i^{\epsilon_n} = \delta_n$ and $Q = \frac{c}{q_1} + \frac{c}{q_2}$. By \myeqref{eq:bstResponse3} we obtain:
\begin{small}
\begin{equation}
\begin{aligned}
  & \sqrt{\frac{Rq_1(\delta_n + x_2^{\epsilon_n})}{c}} - (\delta_n + x_2^{\epsilon_n}) = x_1^{\epsilon_n} \\
  & \sqrt{\frac{Rq_2(\delta_n + x_1^{\epsilon_n})}{c}} - (\delta_n + x_1^{\epsilon_n}) = x_2^{\epsilon_n}. \\
\end{aligned}
\label{eq:perturbed_x1x2}
\end{equation}
\end{small}
Add up the two equations in \myeqref{eq:perturbed_x1x2} and solve it, the positive root is
\begin{small}
\begin{equation}
  x_1^{\epsilon_n} +  x_2^{\epsilon_n} = \frac{1}{2}\big(\frac{R}{Q} + \sqrt{(\frac{R}{Q})^2 + 4\delta_n\frac{R}{Q}}\big) - \delta_n,
\end{equation}
\end{small}
 which will not tend to zero as $n \rightarrow \infty$. This contradicts with the assumption that $\sum_{i=1}^{N}x_i^{\epsilon_n}$ could be infinitely close to zero as $n \rightarrow \infty$.

Finally we prove $\bm{x}^0$ is the equilibrium strategy of the original game by contradiction. Three possible alternative cases need discussing. They are: user $j$ would like to deviate from 
\begin{itemize}
  \item $x_j^0 = 0$ to $x_j^{'} > 0$;
  \item $x_{j}^{0} \in (0,q_j)$ to some $x_j^{'}$ in $[0,q_j] - \{x_j^{0}\}$;
  \item $x_j^0 = q_j$ to $x_j^{'} < q_j$.
\end{itemize}
We just discuss the first case here. Similar method could be applied to prove other cases and we put them in the appendix.  Suppose user $j$ could profitably deviate by contributing a content with quality no less than $\delta(> 0)$.  Then we obtain
\begin{small}
  \begin{equation}
  \centering
  \sqrt{\frac{Rq_jx^{0}_{-j}}{c}} - x_{-j}^{0} \geq \delta.
  \label{eq:originalGNE}
  \end{equation}
\end{small}
Since $x^{0}_{j}=0,$ we know given any sufficiently small positive $\epsilon$, $\exists N^{'}$, $\forall n > N^{'}$, $x_j^{\epsilon_n}<\epsilon$, i.e.,
\begin{small}
  \begin{equation}
  \centering
  \sqrt{\frac{Rq_jx^{\epsilon_n}_{-j}}{c}} - x_{-j}^{\epsilon_n} < \epsilon.
  \label{eq:PerturbGNE}
  \end{equation}
\end{small}
$\bm{x}^{\epsilon_n} \rightarrow \bm{x}^0$ implies $x_{-j}^{\epsilon_n} \rightarrow x_{-j}^{0}$. But as $\epsilon \rightarrow 0$, we can verify $x_{-j}^{\epsilon_n}$ will not converge to $x_{-j}^0$ by \myeqref{eq:originalGNE} and \myeqref{eq:PerturbGNE}. This contradicts with $\bm{x}^{\epsilon_n} \rightarrow \bm{x}^0$. Details could be found in the appendix.

%\begin{equation*}
%\centering
%\frac{1}{2}\big(\frac{Rq_i}{c} - 2\delta - \sqrt{(\frac{Rq_i}{c})^2 - \frac{4\delta Rq_i}{c}}\big) \leq x_{-i}^{0} \leq
%\frac{1}{2}\big(\frac{Rq_i}{c} - 2\delta + \sqrt{(\frac{Rq_i}{c})^2 - \frac{4\delta Rq_i}{c}}\big)
%\end{equation*}

%Finally we prove that when $\epsilon$ is small enough, at least one element in $\bm{x}^\epsilon$ is not equivalent to $\epsilon$. The basic idea is that, if $N-1$ users contribute $\epsilon$, the best response of the $N$-th user will be larger than $\epsilon$. Suppose $x_j=\epsilon, \forall j\ne i$. According to \myeqref{eq:bstResponse3}, the best response of user $i$ is
%$$x_i(q_i  , x_{-i}) = \sqrt{\frac{Rq_i(N-1)\epsilon}{c}} - (N-1)\epsilon.$$
%It is easy to verify that when $\epsilon<\frac{(N-1)Rq_i}{cN^2}$, the best response of $i$ is larger than $\epsilon$. Therefore, $\{x_i=\epsilon\}_{i\in[N]}$ is not an equilibrium of the perturbed game. By setting  $\epsilon \rightarrow 0$, the equilibrium $\bm{x}^\epsilon$ of the perturbed game will converge (but not converge to all-zero contribution vector) to an equilibrium profile of the original game. Thus, PNE of the original game always exists.

Therefore, there exists a PNE for the mechanism $\mathcal{M}_4$.
\end{proof}

%Given the existence of PNE, we can obtain the following lemmas about the properties of an equilibrium profile, which will be used to prove the uniqueness of PNE. The proofs of those lemmas are not difficult, and we put them in the supplementary document due to space limitations.
Given the existence of PNE, we can obtain the following lemmas about the properties of an equilibrium profile, which will be used to find a PNE strategy. The proofs of the first two lemmas can be found in Appendix.
%\begin{lemma}
%In an equilibrium profile, a user with a larger type will contribute no less than a user with a smaller type. Mathematically, if $q_i<q_j$, we have $x_i\le x_j$ in an equilibrium.
%\label{lemma:monotonicyOfbid}
%\end{lemma}
\begin{lemma}
  Consider two users with $q_i\le q_j$. If $x_i=q_i$ holds in an equilibrium, then $x_j=q_j$ holds in the same equilibrium.
  \label{lemma:equiType}
\end{lemma}

Given Mechanism $\mathcal{M}_4$, in which $N$ users compete for the reward $R$, we can induce a \emph{local game} with $m$ users: only the first $m$ users compete for the reward $R$. As shown in the following lemma, an equilibrium of the induced local game can be connected to the equilibrium of the original game (i.e., Mechanism $\mathcal{M}_4$ with $N$ users) under certain condition.

\begin{lemma}
If $\{x_i\}_{i\in[m]}$ is an equilibrium of the induced local game with the first $m$ users and $\sum_{i = 1}^{m}x_i \geq \frac{Rq_{m+1}}{c}$,  then $\{y_i\}_{i\in[N]}$, where $y_i=x_i, \forall 1\le i\le m$ and $y_i=0, \forall m+1\le i\le N$, is an equilibrium of the original game.
\label{lemma:isequi}
\end{lemma}
\begin{proof}
If user $j$($\geq m+1 $) deviates by contributing $z>0$, her utility is
\begin{equation}
\begin{aligned}
u_j &= \frac{Rz}{\sum_{i=1}^{m}x_i + z} - \frac{z}{q_j}c < z(\frac{R}{\sum_{i=1}^{m}x_i} - \frac{c}{q_j}) \\
    & \leq z(\frac{R}{\sum_{i=1}^{m}x_i} - \frac{c}{q_{m+1}}) \\
    & =\frac{cz}{q_{m+1}\sum_{i=1}^{m}x_i}(\frac{Rq_{m+1}}{c} - \sum_{i=1}^{m}x_i) \leq 0.
\end{aligned}
\end{equation}
That is, user $j(j \geq m + 1)$ cannot be better off by unilaterally deviating. The original game is the same as the local game for the first $m$ users, and they will not deviate unilaterally.  Therefore $\{y_i\}_{i\in[N]}$ is an equilibrium of the original game.
\end{proof}

Based on the previous lemmas, we propose Algorithm \ref{alg:searchEqui2} to find a PNE for a local game induced from the original game and verify whether it is a PNE of the original game by Lemma. \ref{lemma:isequi}. Then we discuss how to find a PNE for mechanism $\mathcal{M}_4$ using Algorithm \ref{alg:searchEqui2}.
%\renewcommand{\algorithmicrequire}{\textbf{Input:}}
%\renewcommand{\algorithmicensure} {\textbf{Output:} }
%\begin{algorithm}[!htb]
%\caption{The algorithm to find a PNE for an induced local game}
%\begin{algorithmic}[1]
%\REQUIRE~~ $\bm{q}=(q_1, q_2 , \ldots q_n)$ where $q_1 \geq q_2 \geq \ldots \geq q_n$;\\
%\ENSURE~~ $\bm{x}^{(n)}=(x_1^{(n)},x_2^{(n)},\ldots,x_n^{(n)})$
%\STATE Calculate the $y_{ni}\;\forall i \in [n]$ with \myeqref{eq:algsyms1}. If none of the $y_{ni}$ is smaller than zero or larger than the corresponding $q_i$, return $\bm{x}^{(n)} \leftarrow \bm{y}_{n}$;\label{step:lessthantype}\
%\FOR{$m \leftarrow 1 : n$}\label{step:for1}
%\STATE Calculate $x_{nim}\;\forall i \in \{m+1,...,n\}$ with \myeqref{eq:algsyms2}. Break this loop until $0 \leq x_{nim} \leq q_i \;\forall i \in \{m+1,...,n\}$ or $m$ equals $n$;\
%\ENDFOR
%\STATE $x^{(n)}_i \leftarrow q_i\;\forall i \in \{1,...m\}$ \\
%       $x^{(n)}_i \leftarrow x_{nim}\;\forall i \in\{m+1...n\}$; return $\bm{x}^{(n)}$
%\end{algorithmic}
%\label{alg:searchEqui2}
%\end{algorithm}

\renewcommand{\algorithmicrequire}{\textbf{Input:}}
\renewcommand{\algorithmicensure} {\textbf{Output:} }
\begin{algorithm}[!htb]
\caption{The algorithm to find a PNE of the original game from an induced local game}
\begin{algorithmic}[1]
\REQUIRE~~ $\bm{q}=(q_1, q_2 , \ldots q_n)$ where $q_1 \geq q_2 \geq \ldots \geq q_n$;\\
\ENSURE~~ $\bm{x}^{(n)}=(x_1^{(n)},x_2^{(n)},\ldots,x_n^{(n)})$
\STATE Calculate the $y_{ni}\;\forall i \in [n]$ with \myeqref{eq:algsyms1}. If none of the $y_{ni}$ is smaller than zero or larger than the corresponding $q_i$, $\bm{x}^{(n)} \leftarrow \bm{y}_{n}$; verify whether it is a PNE of the original game by Lemma. \ref{lemma:isequi}; if so, return $\bm{x}^{(n)}.$\
\FOR{$m \leftarrow 1 : n$}\label{step:for1}
\STATE Calculate $x_{nim}\;\forall i \in \{m+1,...,n\}$ with \myeqref{eq:algsyms2}. $x^{(n)}_i \leftarrow q_i\;\forall i \in \{1,...m\}$,  $x^{(n)}_i \leftarrow x_{nim}\;\forall i \in\{m+1...n\}$ if they are all feasible;\
\IF{$\bm{x}^{(n)}$ is a local PNE(verified by Lemma.\ref{lemma:isfeasible})}
\STATE Verify whether $\bm{x}^{(n)}$ is a PNE of the original game and return it if so;\
\ENDIF
\ENDFOR
\end{algorithmic}
\label{alg:searchEqui2}
\end{algorithm}

\begin{lemma}
In Algorithm \ref{alg:searchEqui2}, $\bm{x}^{(n)}$ is a PNE of the local induced game with $n$ users if the following condition holds:
\begin{itemize}
  \item If $m<n$,
  \begin{small}
  \begin{equation*}
  \frac{Rq_{1}}{2c}(1 - \sqrt{1 - \frac{4c}{R}}) \leq \sum_{k=1}^{n}x_i^{(n)} \leq \frac{Rq_{m}}{2c}(1 + \sqrt{1 - \frac{4c}{R}});
  \end{equation*}
  \end{small}
  \item If $m=n$,
  \begin{small}
  \begin{equation*}
  R \geq \frac{cQ_m^2}{q_i(Q_m - q_i)} \forall i \in \{1,m\}  .
  \end{equation*}
  \end{small}
\end{itemize}
\label{lemma:isfeasible}
\end{lemma}
\begin{proof}
If user $i$ contributes her type in equilibrium, we know the solution of
\begin{small}
  \begin{equation*}
  \sqrt{\frac{Rq_i{x_{-i}}}{c}} - x_{-i} \geq q_i
  \end{equation*}
\end{small}
is non-empty. And we can easily infer that $R \geq 4c$ is a necessary condition for that. We temporarily denote $\sum_{k=1}^{n}x_i^{(n)}$ as $x$.

When $m < n$, $x_j^{(n)} < q_j$ holds $\forall j \in \{m+1,...,n\}$. They could be regarded as the best response strategies without constraints. If anyone of the top $m$ users does not want to change her strategy unilaterally, the following inequality holds $\forall i \in [m]$:
\begin{small}
  \begin{equation}
  \sqrt{\frac{Rq_i(x - q_i)}{c}} - (x - q_i) \geq q_i
  \label{eq:verifyLocalPNE}
  \end{equation}
\end{small}

\myeqref{eq:verifyLocalPNE} suggests that
\begin{small}
\begin{equation*}
\frac{Rq_i}{2c}(1 - \sqrt{1 - \frac{4c}{R}}) \leq x \leq \frac{Rq_i}{2c}(1 + \sqrt{1 - \frac{4c}{R}}) \forall i \in [m],
\end{equation*}
\end{small}
i.e.
\begin{small}
  \begin{equation}
  \frac{Rq_1}{2c}(1 - \sqrt{1 - \frac{4c}{R}}) \leq x \leq \frac{Rq_m}{2c}(1 + \sqrt{1 - \frac{4c}{R}}).
  \end{equation}
\end{small}

When $m = n$, we know all of the $n$ people contribute their types. If nobody could profitably deviate unilaterally, we obtain $\forall i \in [m]$,
\begin{small}
\begin{equation}
    \sqrt{\frac{Rq_i(Q_m - q_i)}{c}} - (Q_m - q_i) \geq q_i,
\end{equation}
\end{small}
i.e.
\begin{small}
\begin{equation}
  R \geq \frac{cQ_m^2}{q_i(Q_m - q_i)} \forall i \in \{1,m\}
\end{equation}
\end{small}
\end{proof}

%\end{small}
%\begin{equation}
%y_{ni} = \frac{R(n-1)}{c\sum_{k =1}^n \frac{1}{q_k}}[1 - \frac{n-1}{q_i\sum_{k=1}^{n}\frac{1}{q_k}}]\label{eq:algsyms1}
%\end{equation}
%\begin{equation}
%Q_{m} = \sum_{i=1}^{m}q_i
%\end{equation}
%\begin{equation}
%A_{nm} = \sum_{i=m+1}^{n}\frac{c}{Rq_i}\qquad\forall n \leq N,n\geq m+1
%\end{equation}
%\begin{equation}
%s_{nm}  = \frac{(n-m-1) + \sqrt{(n - m - 1)^2 + 4Q_mA_{nm}}}{2A_{nm}}
%\end{equation}
%\begin{equation}
%x_{nim} = s_{nm} - \frac{c}{Rq_i}s_{nm}^2 \label{eq:algsyms2}
%\end{equation}
The functions used in Algorithm \ref{alg:searchEqui2} are listed  from \myeqref{eq:algsyms1} to \myeqref{eq:algsyms2} and their derivations are in the appendix.
\begin{align}
&y_{ni} = \frac{R(n-1)}{c\sum_{k =1}^n \frac{1}{q_k}}[1 - \frac{n-1}{q_i\sum_{k=1}^{n}\frac{1}{q_k}}] \label{eq:algsyms1} \\
&Q_{m} = \sum_{i=1}^{m}q_i\\
&A_{nm} = \sum_{i=m+1}^{n}\frac{c}{Rq_i}\qquad\forall n \leq N,n\geq m+1\\
&s_{nm}  = \frac{(n-m-1) + \sqrt{(n - m - 1)^2 + 4Q_mA_{nm}}}{2A_{nm}}\\
&x_{nim} = s_{nm} - \frac{c}{Rq_i}s_{nm}^2 \label{eq:algsyms2}
\end{align}

\begin{theorem}
  Recursively calling Algorithm \ref{alg:searchEqui2} from $n=2$ to $N$, it outputs a PNE of the original game.
\end{theorem}
\begin{proof}
  By Theorem \ref{thm:equiexist_var_por}, the induced local game has a PNE.

% Given an $n \leq N$, if the algorithm stops at Step 1,  $\bm{x}^{(n)}$ is a feasible action profile. Further note that $\bm{x}^{(n)}$ is computed from \myeqref{eq:algsyms1}, which is the best response (without considering feasibility constraint $0\le y_{ni}\le q_i$) strategy; else, according to the PNE structure suggested by Lemma. \ref{lemma:equiType}, the algorithm can traverse all the local PNEs. \delete{Therefore $\bm{x}^{(n)}$ is a PNE of the induced local game.}
%
% The PNE of the original game is also a PNE of some induced local game. Algorithm \ref{alg:searchEqui2} checks the PNEs of all the induced local games. So the PNE of the original game could certainly be found.
% \delete{If the algorithm stops at Step 5,  the output $\bm{x}^{(n)}$ is a feasible action profile. \myeqref{eq:algsyms2} is the best response of user $i$ in the local game given the top $m$ people contributing their types without feasibility constraints. Therefore $\bm{x}^{(n)}$ is a PNE.  Different from the above case, in this equilibrium, the first $m(\leq n)$ users contribute their types (see Lemma \ref{lemma:equiType}) and the next $n-m$ people contribute less than their types (see Lemma \ref{lemma:monotonicyOfbid}).}

 Given an $n \leq N$, if the algorithm stops at Step 1, then $\bm{x}^{(n)}$ is a PNE of the original game.

 If the algorithm does not stop at Step 1, given the top $m$ people contributing their types, \myeqref{eq:algsyms2} could be seen as the strategies that people $m+1,...,n$ do not want to deviate if they are feasible. If the top $m$ people do not want to deviate neither,  we find a PNE of the induced local game. Lemma \ref{lemma:equiType} describes the structure of all the PNEs, and therefore Algorithm \ref{alg:searchEqui2} will traverse all the local PNEs of the induced game.  Further, we note that a PNE of the original game is also a PNE of some induced local game. So the PNE of the original game could certainly be found by Algorithm \ref{alg:searchEqui2}.
\end{proof}

%\delete{We can recursively call Algorithm \ref{alg:searchEqui2} to find a PNE of the original game (Mechanism $\mathcal{M}_1$): Run Algorithm \ref{alg:searchEqui2} with $n=2$. Check whether the output of Algorithm \ref{alg:searchEqui2}  satisfies the condition in Lemma \ref{lemma:isequi}. If yes, we are done; otherwise we run Algorithm \ref{alg:searchEqui2} with $n=3$ and check its output.  Repeating this procedure with increasing $n$, we can find a PNE when the process stops.}

\section{Partial-information Setting}
In this section, we investigate the existence of pure Nash equilibrium of UGC mechanisms under the partial-information setting. We discuss three mechanisms here  and leave another one (because of its difficulty) to the future work.

\subsection{$\mathcal{M}_5$: Top $K$ Allocation, Binary Action Space}
For simplicity, we only consider the case that $R>Kc$ and omit the marginal case $R\le Kc$ here.

Let function $T(x)$ denote the probability that a user with quality $x$ is in one of the top $K$ contributors. Clearly, if $N  - K \leq 0$, we have $T(x) = 1$; when $N  - K \geq 1$, we have
\begin{equation}
    T(x)=\sum_{j=0}^{K-1}\binom{N-1}{j}F(x)^{N-1-j}(1-F(x))^{j}.
    \label{eq:pdf_win}
\end{equation}
Intuitively, a user with higher quality and fewer competitors is more likely to get the reward:
\begin{lemma}
  $T(x)$ is a non-decreasing function of $x$.
  \label{lemma:Tfunction}
\end{lemma}
\begin{proof}
We only need to discuss the non-trivial case, i.e. $N  - K \geq 1$.
\begin{small}
\begin{equation}
  \begin{aligned}
      &  \frac{\partial{T(x)}}{\partial{x}} \\
       =& f(x)\Big\{ (N  - 1)F(x)^{N  - 2} + \sum_{j = N  - 2}^{N  - K}\binom{N  -1}{j}\\
        &\big[ jF(x)^{j - 1}(1 - F(x))^{N  - 1 - j} \\
      - &(N  - 1 - j)F(x)^{j}(1 - F(x))^{N  -2 - j} \big]\Big\}\\
       = &(N  - 1)f(x)\binom{N  - 2}{K - 1}F(x)^{N  - K - 1}(1 - F(x))^{K - 1} \\
       \geq &0
  \end{aligned}
  \label{eq:partialT}
\end{equation}
\end{small}
Thus $T(x)$ is a non-decreasing function of $x$.
\end{proof}

Then we can construct a symmetric cut-off equilibrium \cite{fudenberg1991game} for $\mathcal{M}_5$:
\begin{theorem}
Denote the unique root of the following equation as $x^*$.
\begin{equation}
  %\sum_{n=N-K}^{N-1}\binom{N-1}{n}F(x)^n(1 - F(x))^{N-1-n} = \frac{cK}{R}
  \frac{R}{K}T(x) - c = 0
\end{equation}
$\forall i \in [N]$, we have that
\begin{equation}
  \beta(q_i) = \left\{
  \begin{aligned}
  &q_i & &\textrm{if\;\,}q_i\geq x^*\\
  &0   & &\textrm{if\;\,}q_i < x^*
  \end{aligned}
  \right.
\end{equation}
is an equilibrium strategy of $\mathcal{M}_5$.
\end{theorem}
\begin{proof}
Following the strategy, if a user with type $q(\geq x^*)$ chooses to participate in the game, the probability that she could get the reward is
\begin{small}
\begin{equation}
\begin{aligned}
  P(q) &= \sum_{n = 0}^{N - K - 1}\binom{N-1}{n}F(x^*)^n \\
       & \sum_{j = 0}^{K-1}\binom{N-n-1}{j}(F(q) - F(x^*))^{N-1-n-j}(1-F(q))^j \\
       &+ \sum_{n=N-K}^{N-1}\binom{N-1}{n}F(x^*)^n(1 - F(x^*))^{N-1-n}.
\end{aligned}
\end{equation}
\end{small}
If $q > x^*$, we have $P(q) > \frac{cK}{R}$, and the user's expected utility if she participates in the game is
\begin{small}
  \begin{equation}
    \frac{R}{K}P(q) - c > \frac{R}{K}\frac{cK}{R} - c =0.
  \end{equation}
\end{small}
If $q \leq x^*$, the expected utility for the user is $0$. Thus none can be better off by deviating her strategy unilaterally.
\end{proof}
\subsection{$\mathcal{M}_6$: Top $K$ Allocation, Continuous Action Space}
We first give a general description to a symmetric equilibrium strategy, then prove the existence of PNE when $F$ is the uniform distribution with the method of \cite{krishna2009auction}, and finally generalize the result to any general distribution.

Let us consider a symmetric strategy $\beta()$: each user $i$ with quality $q_i$ will contribute to the site with quality $\beta(q_i)$.
\begin{lemma}
If $\beta()$ is an equilibrium strategy, then $q_i > 0\Rightarrow \beta(q_i)>0$.
\end{lemma}
\begin{proof}
  First we declare that given $\beta()$ is an equilibrium strategy, if $q_i < q_j$ and $\beta(q_i)>0$, then we have $\beta(q_j)>0$. Otherwise, user $j$ can take the action $x_j=\beta(q_i) + \delta$, where $\delta$ is sufficiently small,  to get positive utility, which contradicts with that $\beta()$ is an equilibrium strategy.

  Therefore, if there exists some $q>0$ such that $\beta(q)=0$, then we have $\beta(x)=0 \forall x\in[0,q]$. For any  user whose type falls in $[0,q]$, if she contributes $\epsilon$, then her expected utility is
  \begin{equation}
    \frac{R}{K}\sum_{n =N-K}^{N-1}\binom{N-1}{n}F(q)^n(1-F(q))^{N-n-1} - \frac{\epsilon}{q_i}c
  \end{equation}
We can always find an $\epsilon$ small enough to ensure the equation above is larger than $0$. Therefore $\beta(q_i)$ is not an equilibrium, which leads to a contradiction. Thus, there does not exist a $q>0$ such that $\beta(q)=0$.
\end{proof}

%\delete{
%Let function $T(x)$ denote the probability that a user with quality $x$ is in one of the top $K$ highest qualities. We have
%\begin{equation}
%    T(x)=\sum_{j=0}^{K-1}\binom{N-1}{j}F(x)^{N-1-j}(1-F(x))^{j}.
%    \label{eq:pdf_win2}
%\end{equation}}
Suppose that users $j\neq i$ follow the symmetric equilibrium strategy $\beta()$. If user $i$ pretends that her type is $x$ and contributes $\beta(x)$, her expected utility is
\begin{equation}
  u_i(x;q_i) = \frac{R}{K}T(x) - \frac{\beta(x)}{q_i}c.
\end{equation}
The first order derivative of $u_i$ is
\begin{equation}
  \frac{\partial{u_i(x;q_i)}}{\partial{x}} = \frac{R}{K}\frac{\partial{T(x)}}{\partial{x}} - \frac{c\beta^{'}(x)}{q_i}.
  \label{eq:DT}
\end{equation}

If $\beta(q_i)$ is an equilibrium strategy for user $i$, her expected utility should be maximized at $x=q_i$. That is, we should have
$$\frac{\partial{u_i(x;q_i)}}{\partial{x}}| _{x=q_i} = 0.$$
Note that $\beta(0)=0$. Solving the above equation, we get
\begin{small}
\begin{equation}
  \begin{aligned}
  \beta(x) = & \frac{R(N-1)}{cK}\binom{N-2}{K-1}\int_{0}^{x}tF(t)^{N-K-1}(1-F(t))^{K-1}\mathrm{d}F(t)
  \label{eq:GTOPK}
  \end{aligned}
\end{equation}
\end{small}

Then we have the following results.
\begin{lemma}
If $\beta(x)\le x, \forall x\in[0,1]$, then the function $\beta()$ in the above equation is an equilibrium strategy.
\end{lemma}

However, it is possible that $\beta(x)$ expressed by \myeqref{eq:GTOPK} is larger than $x$. For example, if $F$ is the uniform distribution over $[0,1]$, $\beta(x)$ can be written as below.
\begin{small}
\begin{equation}
  \beta(x) = \frac{R}{cK}(N-1)\binom{N-2}{K-1}\sum_{k=0}^{K-1}(-1)^{K-k-1}\binom{K-1}{k}\frac{x^{N-k}}{N-k}
\end{equation}
\end{small}
Then we have\footnote{$B(,)$ is the beta function. }

\begin{small}
\begin{equation*}
\begin{aligned}
    \beta(1) &= \int_{0}^{1}\frac{Rx}{cK}\frac{\partial T(x)}{\partial{x}}\mathrm{d}x \\
    & =\int_{0}^{1}\frac{R}{cK}(N-1)\binom{N-2}{K-1}\int_0^{1} x^{N-K}(1-x)^{K-1}\mathrm{d}x \\
    & = \frac{R}{cK}(N-1)\binom{N-2}{K-1}B(N-K +1 , K) \\
    & = \frac{R}{cK} \frac{N-K}{N},
\end{aligned}
\end{equation*}
\end{small}
which might be larger than $1$.

If $\beta(x)>x$ for some $x\in[0,1]$, $\beta(x)$ will not be an equilibrium strategy anymore. We need to calibrate $\beta(x)$. For ease of description, we first illustrate how to make calibration when $F$ is the uniform distribution, and then extend to a general distribution.

With some derivations, one can get that the equation $\beta(x) = x$ has at most positive two solutions in the region $(0,1]$ for uniform distribution $F$. If there exist two positive solutions (denote them as $x_1$ and $x_2$, and assume $x_1<x_2$), there will be an $x_p(> x_1)$ that satisfies $\beta^{'} (x_p)= 1$. Then we have:

\begin{theorem}
If $F$ is the uniform distribution over $[0,1]$ and $\beta(x)=x$ has two solutions in $(0,1]$, the following $\beta^*()$ function is an equilibrium, where $x_1$ and $x_p$ are defined above.
  \begin{equation}
    \beta^{*}(x) = \left\{
       \begin{aligned}
       & \beta(x) \; &x \in [0 , x_1] \\
       & x \; &x \in (x_1 , x_p]\\
       & \beta(x) -\beta(x_p) + x_p \; &x \in (x_p , 1] \\
       \end{aligned}
       \right.
  \end{equation}
\end{theorem}
\begin{proof}
First, if $x\in[0 , x_1]$, since $\beta(x) \leq x$, we have that $\beta^*(x)=\beta(x)$ is the best response of type $x$.
%$$\frac{\partial{u(q)}}{\partial{q}} = \frac{c\beta(q)}{q^2} > 0.$$ which shows larger quality will lead to a higher utility.\\

Second, it is clear that the first order derivative  $\beta^{'}(x)$ is larger than $1$ for any  $x \in (x_1 , x_p)$. Suppose that all the others follow strategy $\beta^{*}()$ except user $i$, and suppose she pretends that her type is $x$.
\begin{itemize}
\item
If $x \in (x_1 , x_p)$, we have
$$u_i(x ; q_i) = \frac{R}{K}T(x) - \frac{x}{q_i}c,$$
 and $$\frac{\partial{u_i(x ; q_i)}}{\partial{x}} > \frac{R}{K} \frac{\partial{T(x)}}{\partial{x}} - \frac{\beta^{'}(x)}{q_i}c \geq \frac{R}{K} \frac{\partial{T(x)}}{\partial{x}} - \frac{\beta^{'}(x)}{x}c= 0.$$ Therefore, the larger  $x$ is, the larger utility she will get. However, since the contributed quality is upper bounded by her type $q_i$, the best choice for her is to take the action $x_i=q_i$.
 \item If $x\in[0 , x_1]$, we have
$$\frac{\partial{u_i(x ; q_i)}}{\partial{x}} = \frac{R}{K} \frac{\partial{T(x)}}{\partial{x}} - \frac{\beta^{'}(x)}{q_i}c > \frac{R}{K} \frac{\partial{T(x)}}{\partial{x}} - \frac{\beta^{'}(x)}{x}c = 0.$$
So she should pretend her type is $x_1$, which is still worse than revealing the true type $q_i$.
\end{itemize}
Thus, for any $x$ in $(x_i, x_p]$, the best response is $\beta^{*}(x) = x$.

Third, note that for any $x$ in $(x_p , 1]$, we have $\beta^{'}(x) \leq 1$. Integrating $\beta^{'}(x)$ from $x_p$ to $x$ and using $\beta(x_p) = x_p$, we get
$$\beta^{*}(x) - x_p = \beta(x) - \beta(x_p).$$
It is easy to verify that $\beta^{*}(x)\le x$ for any $x$ in $(x_p,1]$. Therefore, we get that $\beta^{*}(x) = \beta(x) + x_p - \beta(x_p)$ is the best response for any $x$ in $(x_p , 1]$.

Thus, the theorem is proved.
\end{proof}

Figure \ref{fig:egusrS} shows an equilibrium strategy for $N = 11 , K = 5 , c = 1 , R = 8$.
\begin{figure}[!htb]
\centering
\includegraphics[width=0.8\columnwidth]{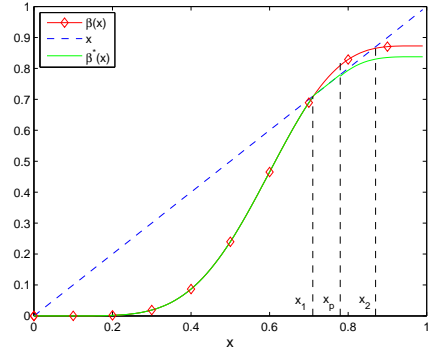}
\caption{An example equilibrium strategy}
\label{fig:egusrS}
\end{figure}

Next we generalize the above results. For a general distribution $F$ over $[0,1]$, we first initialize $\beta^*(x)=\beta(x), \forall x\in[0,1]$ and then calibrate $\beta^*(x)$ as follows.
\begin{enumerate}
  \item Check whether $\beta^*(x)>x$ starting from $x=0$ to $x=1$.
  \item Suppose $[x_1, x_2]$ is the first interval that $\beta^*(x)>x$, and $x_p$ is the point in this interval satisfying $\beta^{'*}(x_p)=1$. Let $o$ denote the value of $\beta^{*}(x)$ at $x_p$ (i.e., $o=\beta^*(x_p)$), and then calibrate $\beta^*(x)=x, \forall x\in[x_1,x_p]$ and $\beta^*(x)= \beta^*(x)-o+x_p, \forall x\in(x_p,1]$.
\item Continue to check whether $\beta^*(x)>x$ starting from $x=x_p$ to $x=1$. If there is still some interval with $\beta^*(x)>x$, we calibrate $\beta^*(x)$ as shown in Step 2.
\item We repeat the checking and calibrating procedure until $\beta^*(x)\le x, \forall x\in[0,1]$.
\end{enumerate}

After this calibration process, we obtain an equilibrium strategy $\beta^*(x)$ from $\beta(x)$, which is shown in \myeqref{eq:GTOPK}, for any distribution $F$. Therefore we have the following theorem.
\begin{theorem}
$\mathcal{M}_6$ has at least one symmetric PNE.
\end{theorem}

\subsection{$\mathcal{M}_7$: Proportional Allocation, Binary Action Space}
Now we study the existence of PNE of the mechanism with the proportional allocation rule and the binary action space under the partial-information setting.

For user $i$, let us consider the following cut-off strategy:
\begin{equation} \label{eq:7beta}
 \beta_i(x) = \left\{
  \begin{aligned}
     & q_i & \text{if } x\ge x^* \\
     & 0 & \text{if } x<x^*,
  \end{aligned}
  \right.
\end{equation}
where $x^*$ is a threshold parameter.

Suppose that users $j\ne i$ follow the above strategy. Then the expected utility of user $i$ can be written as follows if she participate in the game ($x_i=q_i$).
    {\small\begin{equation} \label{eq:7util}
  \begin{aligned}
 u_i(q_i;x^*) &= \sum_{k=0}^{N-1}\binom{N-1}{k}F(x^*)^{N-1-k}(1-F(x^*))^ku_i(q_i,k;x^*) - c\\
        & \overset{def}{=} y(q_i , x^*) - c,
        %&\triangleq y(q_i , x^*) - c,
  \end{aligned}
\end{equation}
}
where $u_i(q_i,k;x^*)$ is the expected utility of user $i$ given another $k$ users with quality larger than $x^*$ participating in the game, and it can be written as
%$$u_i(q_i,k;x^*)=R\int_{\bm{x} \in \bm{D}^k}\frac{q_i}{q_i + x_1 + x_2 + \ldots + x_k}\mathrm{d}F(\bm{x}|x^*),$$
\begin{equation*}
  %u_i(q_i,k;x^*)=
  R\int_{x^*}^{1}...\int_{x^*}^{1}\frac{q_i}{q_i + t_1 + t_2 + \ldots + t_k}\mathrm{d}F(t_1|x^*)...\mathrm{d}F(t_k|x^*).
\end{equation*}
In \myeqref{eq:7util}, when $k = 0$, $u_i(q_i,k;x^*)=R$, which means that user $i$ gets all the reward $R$ since no other user participate in the game ($k=0$ means $x_j=0, \forall j\ne i$).

If \myeqref{eq:7beta} is an equilibrium strategy, then the best response of user $i$ is also to follow the strategy given that all other users follow the strategy. That is,
\begin{equation} \label{eq:7solve}
 u_i(q_i;x^*) = \left\{
  \begin{aligned}
     & >0 & \text{if } q_i\ge x^* \\
     & =0 & \text{if } q_i = x^*\\
     & <0 & \text{if } x<x^*
  \end{aligned}
  \right.
\end{equation}

It is not difficult to get that \begin{itemize}
\item $y(q_i, x^*)$ increases w.r.t. $q_i$,
\item $y(0,0)-c=-c<0$ and $y(1,1)-c=R-c>0$.
\end{itemize}
We further assume that $F(.|.)$ is a continuous function; consequently, $y(t,t)$ is continuous. Therefore, there exists an $x^*$ satisfying the three conditions in \myeqref{eq:7solve}; in turn, this $x^*$ makes \myeqref{eq:7beta} a (symmetric) equilibrium strategy. Thus we have the following theorem.
\begin{theorem}
$\mathcal{M}_7$ has at least one PNE if $R>c$.
\end{theorem}

\section{Conclusions and Future work}
We studied UGC mechanisms under a new framework: Users are heterogeneous and the best quality a user can contribute can be different from others. Under the framework, we considered several mechanisms involving two allocation rules, two action spaces and two information settings. We proved the existence of multiple PNE for some mechanisms, the existence and uniqueness of PNE for some mechanisms, and the non-existence of PNE for some other mechanisms.

There are many issues to explore about UGC mechanisms in the future. First, the efficiency analysis is a  meaningful topic given the existence of multiple equilibria for some mechanisms. Second, we plan to study the mixed Nash equilibrium for UGC mechanisms. Third, we only considered linear cost function in this work. We will investigate more general cost functions (e.g., concave functions). Fourth, the comparison between different mechanisms would be an interesting topic.

\bibliographystyle{aaai}
\bibliography{mybib}

%\clearpage
\appendix{\textbf{Appendix}}

\section{Omitted Proofs}
In this section, we give some technical details that are omitted in the main paper.
%\subsection{Proof of Lemma \ref{lemma:monotonicyOfbid}}
%\begin{proof}
%The proof is by contradiction.
%If the statement was false, there are two possible cases: \\
%\textbf{Case 1}: $q_i < q_j$, $x_i > x_j$, and $x_i \leq q_i, x_j \leq q_j$\\
%Since $x_i$ and $x_j$ are the best responses of $i$ and $j$ separately, we have
%\begin{equation*}
%  \left\{
%  \begin{aligned}
%    & x_i = \sqrt{\frac{R q_i (b + x_j)}{c}} - (b + x_j) \\
%    & x_j = \sqrt{\frac{R q_j (b + x_i)}{c}} - (b + x_i)
%  \end{aligned}
%  \right.
%\end{equation*}
%where $b = \sum_{k \ne i , j}x_k$. Then we get
%\begin{equation}
% q_i(b + x_j) = q_j(b + x_i).  \label{eq:errorInfer}
%\end{equation}
%  Note that $q_i < q_j$ and $x_i > x_j$ lead to $q_i(b + x_j) < q_j(b + x_i)$, which conflicts with \myeqref{eq:errorInfer}.\\
%\textbf{Case 2}: $x_j < q_i = x_i < q_j$:\\
%Similarly, we have
%\begin{equation}
%  \left\{
%  \begin{aligned}
%     & \sqrt{\frac{Rq_i(b+x_j)}{c}} - (b + x_j) \geq q_i \\
%     & \sqrt{\frac{Rq_j(b+x_i)}{c}} - (b + x_i) =    x_j
%  \end{aligned}
%  \right.
%\end{equation}
%and
%$$q_i(b+x_j) \geq q_j(b + x_i).$$
%This clearly contradicts with $ q_i(b+x_j) < q_j(b + x_i).$
%\end{proof}

\subsection{Proof of Lemma \ref{lemma:equiType}}
\begin{proof}
The proof is by contradiction.
 Suppose user $i$ has a larger type than user $j$ ($q_i \geq q_j$) and contributes $x_i < q_i$ in equilibrium.

Setting $A = \sum_{k \ne i,j}x_k$, we have
  \begin{equation}
    \sqrt{\frac{Rq_i(A + q_j)}{c}} - (A + q_j) = x_i
    \label{eq:lemma21}
  \end{equation}
and
  \begin{equation}
    \sqrt{\frac{Rq_j(A + x_i)}{c}} - (A + x_i) \geq q_j.
    \label{eq:lemma22}
  \end{equation}
  From \myeqref{eq:lemma21} and \myeqref{eq:lemma22} we get
  \begin{equation}
    q_j(A + x_i) \geq q_i(A + q_j).
  \end{equation}
  With some simple derivations we can see
  \begin{equation}
    (q_j - q_i) A \geq q_j(q_i - x_i).
    \label{eq:lemma23}
  \end{equation}
  The left hand side of \myeqref{eq:lemma23} is smaller than zero but the right hand side is larger than zero. It is a contradiction.
\end{proof}

\subsection{Derivation of \myeqref{eq:algsyms1}}
By summing \myeqref{eq:foc_utility} over $i$, we get
\begin{equation}
  x_i + x_{-i} = \frac{R(N-1)}{c\sum_{k=1}^{N}\frac{1}{q_k}}.
  \label{eq:sum_interior}
\end{equation}
Substituting the above equation to \myeqref{eq:bstResponse3}, we obtain
\begin{equation}
   x_i=\frac{R(N-1)}{c\sum_{k =1}^N \frac{1}{q_k}}[1 - \frac{N-1}{q_i\sum_{k=1}^{N}\frac{1}{q_k}}],
   \label{eq:bstresponse_interior}
\end{equation}
which is the best response without considering feasibility constraints.

\subsection{Derivation of \myeqref{eq:algsyms2}}
Consider the case that only the first $n$ users participating in the game and the first $m$ users out of the $n$ users contribute their types. For user $ i \in \{m + 1 , m+2...n\}$,  \myeqref{eq:foc_utility} can be written as
\begin{equation}
  \frac{R(Q_m + y_{-i})}{(Q_m + y_i + y_{-i})^2} = \frac{c}{q_i}\quad\quad \forall i \in \{m + 1 , m+2...n\}.
\end{equation}
Summing the above equation over $i$, we obtain
\begin{equation}
  R\frac{(n - m)(Q_m + y_i + y_{-i}) - \sum_{i = m+1}^{n}y_i}{(Q_m + y_i + y_{-i})^2} = \sum_{i = m + 1}^{n}\frac{c}{q_i}.
\end{equation}
Solving the above equation, we get
$$y_i + y_{-i}=$$
$$\frac{(n-m-1) - 2Q_mA_{nm} \pm \sqrt{(n - m - 1)^2 + 4Q_mA_{nm}}}{2A_{nm}}.$$
We are only interested in the positive solution.
Denoting
\begin{equation}
\begin{aligned}
 & s_{nm} = Q_m + \sum_{i = m + 1}^{n}y_i \\
         &= \frac{(n-m-1) + \sqrt{(n - m - 1)^2 + 4Q_mA_{nm}}}{2A_{nm}},
\end{aligned}
\end{equation}
we arrive at
\begin{equation}
  y_i = s_{nm} - \frac{c}{Rq_i}s_{nm}^2.
\end{equation}

%\subsection{Calculation of $\frac{\partial T(x)}{\partial x}$ in \myeqref{eq:DT}}
%By setting $n=0$ of $\frac{\partial T(x,n)}{\partial x}$ in Lemma \ref{lemma:Tfunction}, we can get $\frac{\partial T(x)}{\partial x}$.

\subsection{Omitted proof of Theorem 1}
Define
\begin{equation}
  l(\delta , q) =  \frac{1}{2}(\frac{Rq}{c} - 2\delta - \sqrt{(\frac{Rq}{c})^2 - 4\delta\frac{Rq}{c}})
\end{equation}
and
\begin{equation}
  u(\delta , q) = \frac{1}{2}(\frac{Rq}{c} - 2\delta + \sqrt{(\frac{Rq}{c})^2 - 4\delta\frac{Rq}{c}}).
\end{equation}
We find that
\begin{equation}
\begin{aligned}
  l(\delta , q) &= \frac{1}{2}(\frac{Rq}{c} - 2\delta - \sqrt{(\frac{Rq}{c})^2 - 4\delta\frac{Rq}{c}}) \\
  & =  \frac{2\delta^2}{\frac{Rq}{c} - 2\delta + \sqrt{(\frac{Rq}{c})^2 - 4\delta\frac{Rq}{c}}}.
\end{aligned}
\end{equation}
So $l(\delta , q)$ monotonously increases with $\delta$ and $u(\delta , q)$ monotonously decreases with $\delta$.

If \myeqref{eq:originalGNE} has no solution, we know that she will not change her strategy to generate a content with quality larger than $\delta$. We only consider the case that \myeqref{eq:originalGNE} is solvable. 
\begin{small}
\begin{equation}
\begin{aligned}
   & \sqrt{\frac{Rq_jx_{-j}^0}{c}} - x^0_{-j} \geq \delta \Rightarrow l(\delta , q_j) \leq x_{-j}^0 \leq u(\delta , q_j)\\
   & \sqrt{\frac{Rq_jx_{-j}^{\epsilon_n}}{c}} - x^{\epsilon_n}_{-j} < \epsilon \Rightarrow x_{-j}^{\epsilon_n} < l(\epsilon , q_j)\textrm{ or }x_{-j}^{\epsilon_n} > u(\epsilon , q).
\end{aligned}
\end{equation}
\end{small}
As $\epsilon \rightarrow 0$, $\exists \delta_1 > 0, \delta_2 > 0$, s.t. 
\begin{equation}
  \begin{aligned}
  & l(\delta , q_j) - l(\epsilon , q_j) \geq \delta_1 \\
  & u(\epsilon , q_j) - u(\delta , q_j) \geq \delta_2,
  \end{aligned}
\end{equation} 
which shows that $x_{-j}^{\epsilon_n}$ does not converge to $x_{-j}^0$. It is contradicted with $\bm{x}^{\epsilon_n} \rightarrow \bm{x}^0$.

If $0 < x^0_j < q_j$, we know $x_{-j}^{\epsilon_n}$ will converge to $l(x_j^0 , q_j)$ or $u(x_j^0 , q_j)$. If user $j$ wants to deviate to
\begin{itemize}
  \item some another $x_j^{'} \in (0,q_j)$ but $x^0_j \ne x_j^{'}$, we know  $x_{-j}^{0}  = l(x_j^{'} , q_j)$ or $x_{-j}^{0} = u(x_j^{'} , q_j)$, which contradicts with the fact that the series $\{\bm{x}^{\epsilon_n}\}$ converges.
  \item $x_j^{'} = q_j$, we obtain $l(q_j,q_j) \leq x_{-j}^{0} \leq u(q_j,q_j)$. By the monotonicity of $l(,)$ and $u(,)$ we can find a contradiction.
  \item $x_j^{'} = 0$, we obtain $x_{-j}^{0} \leq l(0,q_j)$ or $x_{-j}^{0} \geq u(0,q_j)$. Neither $l(x_j^0 , q_j)$ nor $u(x_j^0 , q_j)$ could fall in those regions.
\end{itemize}

If $x^0_j = q_j$ but she wants to deviate to some $\delta^{'} \leq \delta < q_j$, we know $x_{-j}^{0} \leq l(\delta , q_j)$ or $x_{-j}^0 \geq u(\delta , q_j)$. While $x^0_j = q_j$ suggests that $x_{-j}^{\epsilon_n}$ will converge to a point in the region $l(q_i,q_i) \leq x_{-j}^0 \leq u(q_i,q_i)$, so we find a contradiction.

\end{document}